\documentclass[aps,pra,reprint,superscriptaddress,amsmath,amssymb,bbm,floatfix]{revtex4-2}

\usepackage[utf8]{inputenc}
\usepackage[T1]{fontenc}
\usepackage{color}
\usepackage{amstext}
\usepackage{graphicx}
\usepackage[colorlinks, allcolors={blue}]{hyperref}
\usepackage{physics}
\usepackage{blindtext}
\usepackage{pifont}
\usepackage{xcolor,graphicx}
\usepackage{newlfont}
\usepackage{amssymb,amsmath,mathrsfs,amsthm}
\usepackage{verbatim}
\usepackage{bbm}
\usepackage{multirow}
\usepackage[varg]{txfonts}
\usepackage{quantikz}

\newcommand{\orcid}[1]{\href{https://orcid.org/#1}{\includegraphics[width=7pt]{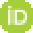}}}

\renewcommand\bra[1]{{\langle{#1}|}}
\renewcommand\ket[1]{{|{#1}\rangle}}
\renewcommand{\Tr}{\text{Tr}}
\newcommand{\I}{\mathbbm{1}}

\newcommand{\sgn}{\mathrm{sgn}}

\newtheorem{theorem}{Theorem}

\theoremstyle{remark}

\allowdisplaybreaks

\begin{document}

\title{Robust self-testing of a maximally entangled qubit pair with arbitrary binary measurements}

\author{Alexandre C. Orthey, Jr.\orcid{0000-0001-8111-3944}}
\email{alexandre.orthey@gmail.com}
\affiliation{Faculty of Mathematics, Informatics and Mechanics, University of Warsaw, ulica Banacha 2, 02-097 Warsaw, Poland}
\affiliation{Instituto de F\'isica Te\'orica, Universidade Estadual Paulista, Rua Dr. Bento Teobaldo Ferraz, 271,
Bloco II, CEP 01140-070 S\~ao Paulo, S\~ao Paulo, Brazil}

\author{Magdalena Stobi\'nska-Moretto\orcid{0000-0002-5168-433X}}
\affiliation{Faculty of Mathematics, Informatics and Mechanics, University of Warsaw, ulica Banacha 2, 02-097 Warsaw, Poland}

\date{\today}

\begin{abstract}
Robust self-testing certifies quantum states and measurements from Bell correlations in the presence of experimental noise. Many analytic proofs assume projective measurements; extending their guarantees to arbitrary measurements requires additional control over the action of the original measurement operators. We provide a direct treatment of this problem in the CHSH scenario, establishing a robust self-test of a maximally entangled qubit pair and Pauli observables for arbitrary, possibly mixed, bipartite states and binary positive operator-valued measurements on both sides. By bounding the change in measurement action under regularization, we derive explicit analytic $\mathcal{O}(\sqrt{\epsilon})$ vector-norm bounds for the state and the joint action of the original observables, where $\epsilon$ is the deficit from the maximal CHSH value. The same local isometries and normalized auxiliary vector apply to all settings, and the isometries leave an external purifying system untouched. We compare these guarantees with projective self-tests and general lifting results in a common error metric, and discuss the limitations imposed by the prefactor in a photonic Bell scheme.
\end{abstract}

\maketitle

\section{Introduction}
The development of quantum technologies requires rigorous certification methods \cite{Eisert2020Quantum}. Among these, device-independent (DI) certification stands out as the strongest form of verification, allowing for the characterization of quantum systems based solely on the classical input-output statistics of an experiment, without requiring any assumptions about the internal workings of the devices \cite{Brunner2014Bell}. The most powerful DI certification technique is known as \textit{self-testing}, originally introduced by Mayers and Yao \cite{Mayers2004Self}, which guarantees that up to local isometries and extra degrees of freedom, a specific quantum state and set of measurements are the only ones capable of producing a given maximal violation of a Bell inequality. An excellent review on the subject was done by \v{S}upi\'{c} and Bowles \cite{Supic2020Sep}.

Over the years, self-testing has been extensively developed and applied to a wide variety of pure entangled states and projective measurements \cite{Reichardt2013Apr,McKague2011TQC,Wu2014Oct,Wang2016Feb,Supic2016Mar,Coladangelo2017May,Kaniewski2019Oct,Mancinska2024Sep,Tavakoli2021Feb,Sarkar2021Oct,Sarkar2022Mar,Supic2023Feb,Chen2024Oct,Sarkar2025May,Sarkar2026Mar,Balanzo-Juando2026Mar}. Because realistic experiments never achieve perfect maximal violations due to unavoidable noise, standard frameworks have been extended to \textit{robust self-testing} \cite{McKague2012Oct, Yang2013May,Bamps2015May,Kaniewski2016Analytic,Kaniewski2020Sep,Sarkar2025Mar}, which bounds the distance between the physical implementation and the ideal reference experiment as a function of the observed experimental error $\epsilon$.

Many robust self-testing proofs are formulated for projective measurements, whose binary observables are Hermitian and unitary. In the Clauser--Horne--Shimony--Holt (CHSH) scenario, Bamps and Pironio~\cite{Bamps2015May} obtained square-root robustness for the state and the joint action of projective observables. Naimark's theorem represents arbitrary positive operator-valued measures (POVMs) as projective measurements on enlarged local Hilbert spaces while preserving the observed correlations. Transferring a guarantee about the action of the dilated observables to the original measurement operators requires an additional argument. Baptista et al.~\cite{Baptista2025Nov} provide general robust lifting results that remove projectivity and purity assumptions under suitable conditions on the reference strategy. Thus, arbitrary POVMs are already covered by general self-testing theory. Kaniewski~\cite{Kaniewski2017Jun} also works directly with binary POVM observables and obtains robust bounds on their effective incompatibility.

General POVMs also appear in semi-DI certification schemes. For instance, Tavakoli et al. \cite{Tavakoli2020Apr} certify non-projective measurements in a prepare-and-measure scenario with a bounded Hilbert space dimension. In a steering scenario, classes of $d$-outcome measurements can be certified with one party's measurements trusted and the other party's measurements arbitrary \cite{Orthey2025Jun}. Here, we consider the fully DI bipartite Bell scenario.

In this work, we give a direct robust self-test of the maximally entangled state of two qubits and Pauli observables, allowing arbitrary mixed physical states and binary POVMs on both sides. We obtain explicit $\mathcal{O}(\sqrt{\epsilon})$ vector-norm bounds for the state and the joint action of the original observables, with the same local isometries and normalized auxiliary vector for every setting. The isometries act trivially on an external purifying system. We use a modified sign function to regularize Alice's observables $A_i$ and Bob's CHSH combinations $\widetilde{B}_i=(B_0+(-1)^i B_1)/\sqrt{2}$ for $i\in\{0,1\}$. The two regularized observables on each side approximately anti-commute on the purification. For each $i$, the vectors obtained by applying the regularized $A_i$ and $\widetilde{B}_i$ separately to the purification are close in norm. This is the consistency condition required in Theorem~1 of McKague et al.~\cite{McKague2012Oct}.

For each pair of settings, we compare the vector obtained by applying $A_x\otimes\widetilde{B}_y$ to the purification with the one obtained using the regularized operators. Their distance, $\Delta$, is preserved by the extraction isometry. The triangle inequality then gives the upper bound $\Delta+f_\mathrm{unitaries}$, where $f_\mathrm{unitaries}$ bounds the extraction error for the regularized operators. We estimate the two terms separately from the observed CHSH deficit $\epsilon=2\sqrt{2}-\expval{\mathfrak{B}}$ and obtain an $\mathcal{O}(\sqrt{\epsilon})$ bound for each.

We also discuss the limitations of this bound for photonic Bell experiments. The displaced on-off detection scheme with a heralded path-entangled single photon studied by Caprara Vivoli et al. \cite{Vivoli2015Comparing} reaches a CHSH value of approximately $2.69$ in its ideal limit. At this value, our bound is uninformative because of its large prefactor. This does not exclude useful analytical state-certification bounds \cite{Kaniewski2016Analytic} or numerical certification methods \cite{Bancal2011Aug}. In particular, the projective representation used in the Navascu\'es--Pironio--Ac\'in (NPA) hierarchy remains valid for correlations generated by arbitrary POVMs \cite{Navascues2008Jul}.

\section{Fundamental concepts}
Let us revisit the notion of self-testing \cite{Supic2020Sep} in a standard Bell test with two parties, Alice and Bob, each having two inputs and two outputs. In this scenario, Alice and Bob share an unknown bipartite quantum state $\rho_\mathcal{AB}$ acting on $\mathcal{H}_\mathcal{A}\otimes\mathcal{H}_\mathcal{B}$ and perform unknown local measurements $\{A_x\}$ and $\{B_y\}$, where $x$ and $y$ label the measurement choices. The only information available to the experimenters is the observed probability distribution $\mathbf{p}=\{p(a,b|x,y)\}$, which we call the \textit{behavior}.

Using Born's rule, we can write
\begin{equation}
  p(ab|xy)=\Tr \left(A_{a|x}\otimes B_{b|y}\rho_\mathcal{AB}\right),
\end{equation}
where $a,b\in\{-1,+1\}$, and $A_{a|x}\geqslant0$ and $B_{b|y}\geqslant0$ are the measurement elements satisfying $\sum_a A_{a|x}=\mathbbm{1}$ and $\sum_b B_{b|y}=\mathbbm{1}$ for every measurement choice $x$ and $y$. The corresponding binary observable is $A_x=A_{+1|x}-A_{-1|x}=2A_{+1|x}-\I$, and similarly for Bob. Thus, $A_x^2=\I$ if and only if $A_{+1|x}^2=A_{+1|x}$, which makes it a projector (same conclusion for $A_{-1|x}$). This is why projectivity of a binary measurement is equivalent to unitarity of its Hermitian observable. For arbitrary binary POVMs, we only have $A_x^2\leqslant\I$ and $B_y^2\leqslant\I$.

We keep the shared state $\rho_\mathcal{AB}$ arbitrary and introduce a purification on an external Hilbert space $\mathcal{H}_\mathcal{E}$:
\begin{equation}
\rho_\mathcal{AB}=\Tr_\mathcal{E}\!\left(\ket{\psi_\mathcal{ABE}}\!\bra{\psi_\mathcal{ABE}}\right),
\end{equation}
where $\ket{\psi_\mathcal{ABE}}\in\mathcal{H}_\mathcal{A}\otimes\mathcal{H}_\mathcal{B}\otimes\mathcal{H}_\mathcal{E}$. The purifying system $\mathcal{E}$, with Hilbert space $\mathcal{H}_\mathcal{E}$, is external to Alice and Bob, and all physical measurements and extraction maps act trivially on it. For local extraction maps $U_\mathcal{A}$ and $U_\mathcal{B}$, we define
\begin{equation}
  U=U_\mathcal{A}\otimes U_\mathcal{B},\qquad V=U\otimes\I_\mathcal{E}.
\end{equation}
Throughout, the output tensor factors of $U$ and $V$ are canonically regrouped in the order $\mathcal{A'B'A''B''}$ and $\mathcal{A'B'A''B''E}$, respectively.
For any operator $K_\mathcal{AB}$ on $\mathcal{H}_\mathcal{A}\otimes\mathcal{H}_\mathcal{B}$, we use the expectation convention
\begin{equation}
  \langle K_\mathcal{AB}\rangle
  =\Tr(\rho_\mathcal{AB}K_\mathcal{AB})
  =\langle\psi_\mathcal{ABE}|K_\mathcal{AB}\otimes\I_\mathcal{E}|\psi_\mathcal{ABE}\rangle.
\end{equation}
Purification and Naimark dilation are both valid mathematical representations. Purification introduces a system on which the local extraction maps must act trivially. Naimark dilation instead introduces local ancillary systems and projective representatives; it does not imply that the original observables satisfy $A_x^2=B_y^2=\I$. Here, we work directly with the original observables and explicitly construct the local extraction maps.

Next, consider a \textit{reference experiment} that reproduces the same behavior $\mathbf{p}$. In this reference realization, Alice and Bob share a known pure state $\ket{\psi'_\mathcal{A'B'}}\in\mathcal{H}_\mathcal{A'}\otimes\mathcal{H}_\mathcal{B'}$ of known dimension, and perform known local extremal measurements $\{A_x'\}$ and $\{B_y'\}$. The goal of self-testing is to infer, solely from the observed correlations $\mathbf{p}$, that the \textit{actual} experiment is equivalent to the reference one in the following sense:
\begin{enumerate}
  \item[(i)] There exist auxiliary Hilbert spaces $\mathcal{H}_\mathcal{A''}$ and $\mathcal{H}_\mathcal{B''}$ such that the local Hilbert spaces admit isometric embeddings into $\mathcal{H}_\mathcal{A'}\otimes\mathcal{H}_\mathcal{A''}$ and $\mathcal{H}_\mathcal{B'}\otimes\mathcal{H}_\mathcal{B''}$, respectively.
  \item[(ii)] There exist local isometries $U_\mathcal{A}:\mathcal{H}_\mathcal{A}\to\mathcal{H}_\mathcal{A'}\otimes\mathcal{H}_\mathcal{A''}$ and $U_\mathcal{B}:\mathcal{H}_\mathcal{B}\to\mathcal{H}_\mathcal{B'}\otimes\mathcal{H}_\mathcal{B''}$, with $U_\mathcal{A}^\dagger U_\mathcal{A}=\I_\mathcal{A}$ and $U_\mathcal{B}^\dagger U_\mathcal{B}=\I_\mathcal{B}$,
    such that, with $V=U_\mathcal{A}\otimes U_\mathcal{B}\otimes\I_\mathcal{E}$, the purification satisfies
    \begin{equation}\label{self-testing_state}
      V\ket{\psi_\mathcal{ABE}}
      =\ket{\psi'_\mathcal{A'B'}}\otimes\ket{\xi_\mathcal{A''B''E}},
    \end{equation}
    where $\ket{\xi_\mathcal{A''B''E}}\in\mathcal{H}_\mathcal{A''}\otimes\mathcal{H}_\mathcal{B''}\otimes\mathcal{H}_\mathcal{E}$ is a normalized auxiliary state that may be entangled across the bipartition $\mathcal{A''B''|E}$, and the joint measurements act on the purification such that
    \begin{equation}\label{self-testing_measurements}
      V (A_x \otimes B_y\otimes\I_\mathcal{E}) \ket{\psi_\mathcal{ABE}} = (A_x' \otimes B_y') \ket{\psi'_\mathcal{A'B'}} \otimes \ket{\xi_\mathcal{A''B''E}},
    \end{equation}
    for every $x$ and $y$. We also include formal identity choices on either side, with the corresponding reference operator equal to the identity, so that Eq.~\eqref{self-testing_measurements} covers the state and individual observable actions. For any purification of $\rho_\mathcal{AB}$, the same local isometries and a single auxiliary vector (which may depend on the purification, but not on $x$ or $y$) satisfy Eqs.~\eqref{self-testing_state} and \eqref{self-testing_measurements} simultaneously.
\end{enumerate}
If conditions (i) and (ii) are satisfied, we say that the reference state $\ket{\psi'_\mathcal{A'B'}}$ and the reference measurements $\{A_x',B_y'\}$ are \textit{self-tested} by the behavior $\mathbf{p}$. %It is crucial to define this measurement equivalence via its action on the state rather than as a strict global operator equivalence, such as $U A_x U^\dagger = A'_x \otimes \I$. Because unitary transformations preserve the spectrum of an operator, global equivalence would mathematically force the physical observables to be strictly projective. Defining it via the state allows the physical $A_x$ to be a non-unitary POVM, provided its action on the support of the state mimics the ideal reference.

This formulation is sufficient for the real reference state and observables considered here. For general complex reference measurements, self-testing must allow for the complex-conjugation ambiguity, since the behavior $\mathbf{p}$ is invariant under simultaneous complex conjugation of the state and measurement operators \cite{Supic2020Sep} (see also Ref. \cite{Chen2025Dec} for a deeper discussion on the role of complex numbers in self-testing).

To check whether a certain behavior $\mathbf{p}$ implies conditions (i) and (ii), we can use a non-trivial Bell-type inequality $\expval{\mathfrak{B}}\leqslant \beta_C$, where $\beta_C$ is its classical bound and the Bell operator $\mathfrak{B}$ is written in terms of the measurements that can be performed by Alice and Bob,
\begin{equation}
  \mathfrak{B}\coloneqq \sum_{i,j} \gamma_{ij}A_i\otimes B_j,
\end{equation}
with suitable coefficients $\gamma_{ij}$. The inequality $\expval{\mathfrak{B}}\leqslant \beta_C$ is tailored such that it can only be maximally violated by the reference measurements and state, up to local isometries and extra degrees of freedom. Thus, the self-testing statement is a theorem guaranteeing that if the quantum bound $\beta_Q>\beta_C$ is achieved, then conditions (i) and (ii) hold.

\section{Revisiting robust self-testing}
It might happen that the behavior $\mathbf{p}$ does not lead to the maximal violation of the Bell inequality $\expval{\mathfrak{B}}\leqslant \beta_C$, but instead falls short by some $\epsilon>0$, that is,
\begin{equation}\label{robust_I}
  \beta_Q>\expval{\mathfrak{B}}\geqslant \beta_Q -\epsilon,
\end{equation}
such that $\beta_Q-\epsilon > \beta_C$. In that case, one may prove the robustness of the tailored Bell inequality. We say that a behavior $\mathbf{p}$ robustly self-tests the reference measurements $\{A_x',B_y'\}$ and the reference state $\ket{\psi_\mathcal{A'B'}'}$ if, given an $\epsilon>0$ such that \eqref{robust_I} is true, there exists a function $f:\epsilon\mapsto f(\epsilon)$, local isometries $U_\mathcal{A}$ and $U_\mathcal{B}$, and a normalized auxiliary state $\ket{\xi_\mathcal{A''B''E}}$ common to all settings, such that
\begin{align}
  \Big|\Big|& V\left(A_x\otimes B_y\otimes\I_\mathcal{E}\right)\ket{\psi_\mathcal{ABE}} \nonumber\\
  & -A_x'\otimes B_y'\ket{\psi_\mathcal{A'B'}'}\otimes \ket{\xi_\mathcal{A''B''E}}\Big|\Big| <f(\epsilon),\label{norm_robust}
\end{align}
where $V=U_\mathcal{A}\otimes U_\mathcal{B}\otimes\I_\mathcal{E}$ and $\lim_{\epsilon\to 0}f(\epsilon)=0$. Here, $\norm{\,\ket{v}}=\sqrt{\braket{v}{v}}$ is the Euclidean norm. We also include formal identity choices for $A_x$ and $B_y$, so that the bound covers the state and the action of a single local observable. For any purification of $\rho_\mathcal{AB}$, the same local isometries and a single auxiliary vector (which may depend on the purification) work for all settings. For alternative formulations of robust self-testing, see Refs. \cite{Yang2013May,Wu2014Oct,Kaniewski2020Sep,Baccari2020Jan}.

A cornerstone result in this direction is Theorem~1 of McKague, Yang, and Scarani \cite{McKague2012Oct}. Adapted to our notation, the theorem states that if Alice and Bob perform binary measurements with strictly unitary and Hermitian observables (i.e., $A_x^2 = \I $ and $B_y^2 = \I $), one can explicitly bound the distance between the physical implementation and the ideal reference. For a mixed physical state, take any purification $\ket{\psi_\mathcal{ABE}}$ of $\rho_\mathcal{AB}$ and extend all physical operators by $\I_\mathcal{E}$. Specifically, suppose there exist error terms $\epsilon_1$ and $\epsilon_2$ such that the observables' anti-commutators satisfy
\begin{align}\label{eq:mckague_eps1}
  \frac{1}{2}\max\Big\{\,&\norm{\left(\{A_0,A_1\}\otimes\I_\mathcal{B}\otimes\I_\mathcal{E}\right)\ket{\psi_\mathcal{ABE}}},\nonumber\\
  &\norm{\left(\I_\mathcal{A}\otimes\{B_0,B_1\}\otimes\I_\mathcal{E}\right)\ket{\psi_\mathcal{ABE}}}\,\Big\} \leqslant \epsilon_1,
\end{align}
and corresponding observables on Alice's and Bob's sides satisfy the consistency bound
\begin{equation}\label{eq:mckague_eps2}
  \max_{x} \norm{\left[(A_x \otimes \I_\mathcal{B} - \I_\mathcal{A} \otimes B_x)\otimes\I_\mathcal{E}\right]\ket{\psi_\mathcal{ABE}}} \leqslant \epsilon_2.
\end{equation}
Then, there exist local isometries $U_\mathcal{A}:\mathcal{H}_\mathcal{A}\to\mathbb{C}^2\otimes\mathcal{H}_\mathcal{A''}$ and $U_\mathcal{B}:\mathcal{H}_\mathcal{B}\to\mathbb{C}^2\otimes\mathcal{H}_\mathcal{B''}$, and a normalized auxiliary state $\ket{\xi_\mathcal{A''B''E}}$, such that, with $V=U_\mathcal{A}\otimes U_\mathcal{B}\otimes\I_\mathcal{E}$,
\begin{align}\label{eq:mckague_bound}
  &\norm{V\left(A_x\otimes B_y\otimes\I_\mathcal{E}\right)\ket{\psi_\mathcal{ABE}}  -A_x'\otimes B_y'\ket{\phi_\mathcal{A'B'}^+}\otimes \ket{\xi_\mathcal{A''B''E}}} \nonumber \\
  &\leqslant f_\mathrm{unitaries}(\epsilon_1, \epsilon_2),
\end{align}
where $\ket{\phi_\mathcal{A'B'}^+}$ is the ideal pure maximally entangled state, $A_x', B_y'$ are the ideal reference Pauli observables, and the auxiliary state may be entangled across $\mathcal{A''B''|E}$. The same local isometries and a single auxiliary vector (which may depend on the purification) work for every setting, and the statement holds for any purification of $\rho_\mathcal{AB}$. The extraction acts trivially on $\mathcal{E}$; see Appendix~\ref{app:proof} for justification. The robustness bound evaluates to
\begin{equation}\label{f_unitaries_main}
  f_\mathrm{unitaries}(\epsilon_1, \epsilon_2) = \frac{11\epsilon_1 + 5\epsilon_2}{2}.
\end{equation}

The above criterion assumes Hermitian unitary observables, whereas general binary POVMs only give $A_x^2\leqslant\I$ and $B_y^2\leqslant\I$. To use this criterion for the original observables, we apply it to regularized unitary operators and bound the error introduced by this replacement. The following results provide these bounds directly in terms of the observed CHSH violation.

\section{Results}
Let us consider the CHSH Bell operator \cite{Clauser1969Oct}, which is given by
\begin{equation}\label{CHSH_operator}
  \mathfrak{B}\coloneqq A_0\otimes B_0 + A_0\otimes B_1 + A_1\otimes B_0 -A_1\otimes B_1.
\end{equation}
Here, the local physical observables $A_x$ and $B_y$ (for $x,y\in\{0,1\}$) are Hermitian operators originating from binary POVMs, i.e., $A_x = M_{+1|x}-M_{-1|x}$ and $B_y = N_{+1|y}-N_{-1|y}$, such that $\sum_{j\in\{-1,+1\}}M_{j|i}=\sum_{j\in\{-1,+1\}}N_{j|i}=\I $, $M_{j|i}\geqslant 0$, $N_{j|i}\geqslant 0$. Working directly with these original measurement operators, we do not assume they are unitary; instead, they satisfy $A_x^2\leqslant \I$ and $B_y^2\leqslant \I$.

By defining Bob's tilted observables as
\begin{equation}
  \widetilde{B}_0 \coloneqq \frac{B_0+B_1}{\sqrt{2}}\quad\textrm{and}\qquad \widetilde{B}_1\coloneqq \frac{B_0-B_1}{\sqrt{2}},\label{BobsTilted}
\end{equation}
and introducing the operators $P_i \coloneqq \I\otimes\I-A_i\otimes \widetilde{B}_i$, we obtain the identity

\begin{align}\label{SOS_decomposition}
  \frac{1}{\sqrt{2}}\sum_{i=0}^{1} P_i^2
  &=\sqrt{2}\I\otimes \I -\mathfrak{B}+\frac{\sqrt{2}}{4}\left(A_0^2+A_1^2\right)\otimes\left(B_0^2+B_1^2\right)\nonumber\\
  &\quad+\frac{\sqrt{2}}{4}\left(A_0^2-A_1^2\right)\otimes\left\{B_0,B_1\right\}.
\end{align}
Since $A_i^2\leqslant\I$ and $B_i^2\leqslant\I$, we have
\begin{align}
    \sum_{i=0}^1 A_i^2\otimes\widetilde{B}_i^2 & \leqslant\I\otimes(\widetilde{B}_0^2+\widetilde{B}_1^2)\nonumber\\
    &=\I\otimes(B_0^2+B_1^2)\leqslant2\I\otimes\I.
\end{align}
Combining this estimate with $P_i^2\geqslant0$ in Eq.~\eqref{SOS_decomposition} gives $\mathfrak{B}\leqslant2\sqrt{2}\I\otimes\I$, recovering the standard quantum bound $\beta_Q=2\sqrt{2}$. The robustness estimates below are obtained from the Cauchy--Schwarz inequality and vector-norm bounds, as detailed in Appendix~\ref{app:proof}.

Throughout this study, neither purity of $\rho_{AB}$, full rank of its reduced states, nor projectivity of the physical measurements is assumed. Having established this framework, we present our main result.

\begin{theorem}\label{theorem1}
  Let $A_x$ acting on $\mathcal{H}_\mathcal{A}$ and $B_y$ acting on $\mathcal{H}_\mathcal{B}$ be Alice's and Bob's quantum binary observables, respectively, in a CHSH scenario such that $A_x^2\leqslant \I$ and $B_y^2\leqslant \I$ for $x,y\in\{0,1\}$. Let $\rho_\mathcal{AB}$ be an arbitrary density operator on $\mathcal{H}_\mathcal{A}\otimes\mathcal{H}_\mathcal{B}$, and let $\ket{\psi_\mathcal{ABE}}$ be any purification of $\rho_\mathcal{AB}$. If $\Tr(\rho_\mathcal{AB}\mathfrak{B})=2\sqrt{2}-\epsilon$, with $0\leqslant\epsilon\leqslant2\sqrt{2}-2$, then there exist auxiliary Hilbert spaces $\mathcal{H}_\mathcal{A''}$ and $\mathcal{H}_\mathcal{B''}$ and local isometries $U_\mathcal{A}:\mathcal{H}_\mathcal{A}\to\mathbb{C}^2\otimes\mathcal{H}_\mathcal{A''}$ and $U_\mathcal{B}:\mathcal{H}_\mathcal{B}\to\mathbb{C}^2\otimes\mathcal{H}_\mathcal{B''}$, with $V=U_\mathcal{A}\otimes U_\mathcal{B}\otimes\I_\mathcal{E}$, and a normalized auxiliary state $\ket{\xi_\mathcal{A''B''E}}$ such that, for every $x,y\in\{0,1,2\}$,
  \begin{align}
      \Big\|&V\left(A_x\otimes\widetilde{B}_y\otimes\I_\mathcal{E}\right)\ket{\psi_\mathcal{ABE}}\nonumber\\
      &-\left[(A_x'\otimes B_y')\ket{\phi_\mathcal{A'B'}^+}\right]\otimes\ket{\xi_\mathcal{A''B''E}}\Big\|\leqslant C\sqrt{\epsilon}. \label{eq:robust_chsh}
    \end{align}
  Here, $\ket{\xi_\mathcal{A''B''E}}\in\mathcal{H}_\mathcal{A''}\otimes\mathcal{H}_\mathcal{B''}\otimes\mathcal{H}_\mathcal{E}$ may be entangled across the bipartition $\mathcal{A''B''|E}$, and $C=(75+107\sqrt{2})/2^\frac{3}{4}$. The same local isometries and normalized auxiliary vector work for all $x,y\in\{0,1,2\}$, and the reference observables are assigned as
  \begin{align}
    & A_0'=B_0'=\sigma_x,\\
    & A_1'=B_1'=\sigma_z,\\
    & A_2=B_2=A_2'=B_2'=\I,\qquad \widetilde{B}_2=\I_\mathcal{B},
  \end{align}
\end{theorem}

Taking $x=y=2$, the pure-state trace-distance bound and contractivity under partial trace give
\begin{equation}
\frac12\left\|
U\rho_\mathcal{AB}U^\dagger
-\ket{\phi_\mathcal{A'B'}^+}\bra{\phi_\mathcal{A'B'}^+}
\otimes\sigma_\mathcal{A''B''}
\right\|_1
\leqslant C\sqrt{\epsilon},
\end{equation}
where $U=U_\mathcal{A}\otimes U_\mathcal{B}$,
$\sigma_\mathcal{A''B''}=\Tr_\mathcal{E}
(\ket{\xi_\mathcal{A''B''E}}\bra{\xi_\mathcal{A''B''E}})$,
and $\|\varrho\|_1=\Tr\sqrt{\varrho^\dagger\varrho}$ denotes the trace norm. Appendix~\ref{app:comparison_constants} derives the corresponding bounds for the action of Bob's original observables and for the individual local measurement effects, using the same local isometries and auxiliary vector.

\begin{proof}[Sketch of proof]
  The detailed proof is provided in Appendix~\ref{app:proof}. Here we outline the main logical steps.

  \textit{(Step 1. Regularization of observables)} Because the physical observables are not necessarily unitary ($A_x^2 \leqslant \I$, $B_y^2 \leqslant \I$), we cannot directly apply the criterion above. To this end, we regularize Alice's observables via the modified sign function, $\mathbb{A}_x \coloneqq \sgn^*(A_x)$. This function re-assigns non-negative eigenvalues to +1 and negative eigenvalues to -1. For Bob, we directly regularize the original tilted observables $\widetilde{B}_0=(B_0+B_1)/\sqrt{2}$ and $\widetilde{B}_1=(B_0-B_1)/\sqrt{2}$, defining $\widetilde{\mathbb{B}}_0=\sgn^*(B_0+B_1)=\sgn^*(\widetilde{B}_0)$ and $\widetilde{\mathbb{B}}_1=\sgn^*(B_0-B_1)=\sgn^*(\widetilde{B}_1)$; the positive factor $1/\sqrt{2}$ does not affect $\sgn^*$. By construction, these auxiliary observables are Hermitian and exactly unitary, i.e., $\mathbb{A}_x^2=\widetilde{\mathbb{B}}_y^2=\I$.

  \textit{(Step 2. Splitting of the error)} Once the approximate anti-commutation and consistency relations are bounded in Step 4, Theorem 1 of McKague et al.~\cite{McKague2012Oct} guarantees the existence of local isometries $U_\mathcal{A}$ and $U_\mathcal{B}$ and a normalized auxiliary state $\ket{\xi_\mathcal{A''B''E}}$ for the regularized operators. With $V=U_\mathcal{A}\otimes U_\mathcal{B}\otimes\I_\mathcal{E}$, we add and subtract the extracted vector for the regularized operators. Since $V$ preserves the vector norm, the triangle inequality gives
  \begin{align}
    & \norm{V\left(A_x\otimes \widetilde{B}_y\otimes\I_\mathcal{E}\right)\ket{\psi_\mathcal{ABE}}  -\left[(A_x'\otimes B_y')\ket{\phi_\mathcal{A'B'}^+}\right]\otimes \ket{\xi_\mathcal{A''B''E}}}\nonumber\\
    & \leqslant \Delta + f_\mathrm{unitaries}(\epsilon_1, \epsilon_2),\label{eq:error_split}
  \end{align}
  where the deviation term $\Delta$ is given by
  \begin{equation}\label{definition_delta}
    \Delta\coloneqq \norm{\left(\mathbb{A}_x\otimes\widetilde{\mathbb{B}}_y-A_x\otimes\widetilde{B}_y\right)\otimes\I_\mathcal{E}\ket{\psi_\mathcal{ABE}}},
  \end{equation}
  and $f_\mathrm{unitaries}(\epsilon_1,\epsilon_2)$ follows from Eq.~\eqref{f_unitaries_main} using the anti-commutation and consistency bounds in Eqs.~\eqref{eq:mckague_eps1} and~\eqref{eq:mckague_eps2} for the regularized observables.
  
  \textit{(Step 3. Bound for $\Delta$)} Next, we bound $\Delta$ as a function of $\epsilon$. We apply the Cauchy--Schwarz inequality to the vectors introduced in Eqs.~\eqref{CSvectors}--\eqref{CSboundend}, then use the triangle inequality and the parallelogram law of norms to bound the deviations of the physical observables from their regularized counterparts. This yields $\Delta \leqslant 2^\frac{1}{4}\left(2+2\sqrt{2}\right)\sqrt{\epsilon}$.

  \textit{(Step 4. Bound for $f_\textrm{unitaries}$)} Using the estimates from the previous step, we bound terms such as $\norm{(\mathbb{A}_i - \widetilde{\mathbb{B}}_i) \ket{\psi_\mathcal{ABE}}}$, $\norm{\left\{\mathbb{A}_0,\mathbb{A}_1 \right\}\ket{\psi_\mathcal{ABE}}}$, and $\norm{\left\{\widetilde{\mathbb{B}}_0,\widetilde{\mathbb{B}}_1 \right\}\ket{\psi_\mathcal{ABE}}}$, with identity factors on unmentioned systems understood. These give the anti-commutation and consistency bounds required for $f_\mathrm{unitaries}$ in Eq.~\eqref{f_unitaries_main}. We show that $f_\mathrm{unitaries}$ scales as $\mathcal{O}(\sqrt{\epsilon})$, which, upon combining with $\Delta$, yields the overall robustness coefficient $C$.
\end{proof}

\section{Discussion}
\subsection{Implications for photonic Bell experiments}
To illustrate the limitations of our bound, consider the scheme in Fig.~1(B) of Caprara Vivoli et al. \cite{Vivoli2015Comparing}. A heralded photon is split between Alice and Bob, who perform local displacements followed by on-off detection. For this particular scheme, the optimized CHSH value approaches $\expval{\mathfrak{B}}\approx2.69$ in the limit of unit detection efficiency and negligible multiphoton emission, corresponding to $\epsilon\approx0.138$. This value is not a universal ceiling for photonic Bell experiments or on-off detectors.

It is important to distinguish detector resolution from measurement projectivity. An ideal threshold detector has effects $E_0=\ket{0}\bra{0}$ and $E_c=\I-E_0$ on the full Fock space, which are projectors. A unitary displacement preserves projectivity. For efficiency $0<\eta<1$ and no dark counts, the no-click effect is instead $E_0^{(\eta)}=\sum_{n=0}^{\infty}(1-\eta)^n\ket{n}\bra{n}$, which is generally non-projective \cite{Vivoli2015Comparing}. Loss or compression of the displaced measurement to an effective subspace can likewise give non-projective effects. Non-projectivity therefore depends on the measurement model and is not implied by the inability to resolve photon number.

At a CHSH value of $2.69$, our estimate gives $C\sqrt{\epsilon}\approx50$. Since the distance between normalized state vectors is at most two, and the trace distance in the state corollary is at most one, this estimate gives no useful state certification. The large value is a limitation of our upper bound; it does not imply that the experimental state is far from the target. Likewise, $\Delta$ refers to the chosen regularization, and $\Delta+f_\mathrm{unitaries}$ bounds the error of the constructed extraction. Neither this expression nor its upper bound is a lower bound on the error of an optimized extraction or an unavoidable cost of non-projectivity.

Analytical bounds for a different certification task can still be informative at this violation. Kaniewski's state-extractability bound \cite{Kaniewski2016Analytic} gives
\begin{equation}
  F\geqslant\frac{4+5\sqrt{2}}{16}\,\expval{\mathfrak{B}}-\frac{1+2\sqrt{2}}{4},
\end{equation}
where $F$ is the overlap with the maximally entangled two-qubit state, optimized over local extraction channels. At $\expval{\mathfrak{B}}=2.69$, this gives $F\gtrsim0.904$. For arbitrary binary POVMs, local Naimark embeddings can be composed with these channels, so the state bound remains applicable. This state-only fidelity guarantee is different from the simultaneous vector-norm bound for the state and the action of the observables in Eq.~\eqref{eq:robust_chsh}.

Numerical certification methods provide another approach \cite{Bancal2011Aug}. The projective representation used in the NPA hierarchy is valid for arbitrary POVM correlations when the Hilbert space dimension is unrestricted \cite{Navascues2008Jul}; it does not assume that the physical detectors are projective. More generally, noncommutative polynomial optimization can impose positivity and normalization constraints directly on POVM effects \cite{Pironio2010}. A guarantee about the action of the original observables still requires an argument relating them to the operators used in the extraction. Our regularization estimates provide such an argument for CHSH.

\subsection{Relation to earlier CHSH self-tests and general lifting results}\label{sec:comparison}
For projective measurements, Bamps and Pironio~\cite{Bamps2015May} obtained an $\mathcal{O}(\sqrt{\epsilon})$ bound in vector norm for the state and the joint action of Alice's and Bob's observables. Their Eq.~(27) at zero tilt gives the CHSH case, up to exchanging the Pauli $x$ and $z$ axes on both sides. Their Eqs.~(A30) and~(A38) also bound the change on the state when Bob's sum and difference of observables are regularized. Tensoring their local construction with the identity on a purifying system extends these estimates for projective measurements to mixed states. Our construction treats non-projective physical observables on both sides and bounds the change in their action directly.

Kaniewski~\cite{Kaniewski2017Jun} allows Hermitian contractions $A_x^2\leqslant\I$ and $B_y^2\leqslant\I$ on both sides (Sec.~II~A and Appendix~A). His robust CHSH bound controls the effective commutator $\tfrac12\Tr(\rho_\mathcal{A}|[A_0,A_1]|)$, which quantifies the incompatibility of Alice's observables on her reduced state. At maximal violation, his analysis also gives exact rigidity of the state and measurements. The distinction at nonmaximal violation concerns the certification guarantee: here we give explicit vector-norm bounds for the state and observable actions under common local isometries.

Baptista et al.~\cite{Baptista2025Nov} provide a general route from projective measurements to POVMs. Their Theorem~4.3 applies to pure physical states when the reference state has full Schmidt rank, as does our two-qubit maximally entangled target. Their error is a uniform vector-norm bound for the state and for the action of each local measurement effect separately (Definition~2.5). A projective self-test with error bound $r$ gives a POVM self-test for pure states with bound $5r$ at the same Bell value, preserving the robustness exponent (Remark~4.4).

Appendix~\ref{app:comparison_constants} compares the bounds in a common vector-norm metric for the state and individual local measurement effects. For pure physical states, our bound has an approximately $30\%$ smaller prefactor than the composition of the Bamps--Pironio estimates with Baptista et al.'s projectivity lifting. For mixed physical states, our construction retains $\mathcal{O}(\sqrt{\epsilon})$ robustness, while further applying their generic purity-lifting estimate gives an $\mathcal{O}(\epsilon^{1/4})$ bound. These comparisons concern this particular composition of reported estimates, whose constants we do not optimize.

\subsection{Possible extensions}
The local regularization step can be applied to any finite number of binary measurements and parties. For instance, let $T_j$ be a Hermitian contraction on party $j$ and $R_j=\sgn^*(T_j)$. A telescoping expansion, replacing one factor $T_j$ at a time by $R_j$, gives
\begin{equation}
    \left\|\left(\bigotimes_{j=1}^nT_j-\bigotimes_{j=1}^nR_j\right)\ket{\psi}\right\|\leqslant\sum_{j=1}^n\left\|(T_j-R_j)\ket{\psi}\right\|,
\end{equation}
with identity factors on unmentioned systems understood. The other local factors in each term have operator norm at most one, so the regularization error for a product of observables is bounded by the sum of the local errors. To obtain a multipartite self-test or include more inputs, one must also derive the corresponding local deviation and algebraic bounds from the chosen Bell inequality and supply a suitable self-test for the regularized observables. If the local deviation bounds and the resulting unitary self-testing error are $\mathcal{O}(\sqrt{\epsilon})$, the same exponent is retained for a fixed scenario. The coefficient depends on the number of parties, settings, and relations used in the proof; the CHSH result alone gives no universal scaling law.

For more than two outcomes, a single Hermitian unitary operator no longer describes the measurement. Applying the sign function separately to the effects does not ensure their orthogonality or normalization. Such an extension therefore requires a regularization that preserves the measurement structure and a robustness analysis adapted to the chosen reference measurements.

\section{Conclusions}
In this work, we have established a robust self-test for the maximally entangled state of two qubits and Pauli observables in the CHSH scenario. Our explicit vector-norm bounds simultaneously control the state and the joint action of the original observables for arbitrary mixed physical states and binary POVMs. The same local isometries and normalized auxiliary vector apply to every setting, and the isometries act trivially on an external purifying system.

By regularizing the observables with a modified sign function and using the Cauchy--Schwarz inequality and vector-norm estimates, we have bounded their deviation from Hermitian unitary operators on the state. Combining these estimates with the self-test for the regularized observables gives an analytic $\mathcal{O}(\sqrt{\epsilon})$ robustness bound. Its large prefactor limits its usefulness away from the maximal violation.

We also made the quantitative scope of this construction explicit. For pure physical states, converting to individual measurement effects gives a smaller prefactor than the specified composition of the bounds presented by Bamps and Pironio with Baptista et al.'s projectivity lifting. For mixed physical states, our square-root bound improves on the fourth-root bound obtained by further applying their generic purity-lifting estimate. These comparisons concern the stated bounds and do not imply that alternative implementations of the earlier methods must have the same constants or exponent.

For the state alone, both operators in Eq.~\eqref{eq:error_split} are identities and $\Delta=0$, so only $f_\mathrm{unitaries}$ enters the estimate. For the action of the observables, $\Delta$ can contribute even when the physical measurements are projective. Indeed, if $B_0^2=B_1^2=\I$, then $\widetilde{B}_y^2=\I+(-1)^y\{B_0,B_1\}/2$ for $y\in\{0,1\}$, so Bob's CHSH combinations can still require regularization. The separate estimate for $\Delta$ quantifies the change caused by this particular replacement on the state.

A natural next step is to improve the constants while retaining simultaneous control of the state and the action of the observables. It would also be useful to derive the Bell-specific estimates needed for multipartite scenarios and to develop a regularization suitable for measurements with more than two outcomes. Applications to DI quantum key distribution and randomness generation, including photonic implementations \cite{Moradi2026Jun}, would require a separate security analysis using the relevant certification guarantee.

\begin{acknowledgments}
  ACO thanks J\k{e}drzej Kaniewski, Laura Man\v{c}inska, Ranyiliu Chen, Ritopriyo Pal, and the anonymous referees for useful discussions and comments. This work was supported by the QuantERA II Programme that has received funding from the European Union's Horizon 2020 research and innovation programme under Grant Agreement No 101017733, project “PhoMemtor” No. 2021/03/Y/ST2/00177. ACO also acknowledges that this study was financed, in part, by the S\~ao Paulo Research Foundation (FAPESP), Brazil. Process Numbers 2025/10927-1 and 2026/08210-4. M.S.-M. acknowledges support from the Foundation for Polish Science's International Research Agenda Programme project ``Center for Hybrid Quantum-Classical Information Technologies -- QLAB'' No. FENG.02.01-IP.05-B013/25, financed by the FENG programme 2021-2027. The authors used OpenAI's GPT-6 Astra to assist with proof development, language editing, manuscript organization, and bibliographic checking. The authors critically evaluated the AI-generated suggestions and independently verified the final proofs, references, and conclusions, and take full responsibility for the content of the manuscript.
\end{acknowledgments}

\bibliography{bibliography}

\appendix

\setcounter{theorem}{0}

\section{Proof of Theorem 1}\label{app:proof}
For any purification $\ket{\psi_\mathcal{ABE}}$ of $\rho_\mathcal{AB}$, every physical operator has the form of its bipartite counterpart tensored with $\I_\mathcal{E}$. Thus the local extraction maps constructed from these operators extend to the purification as $V=U_\mathcal{A}\otimes U_\mathcal{B}\otimes\I_\mathcal{E}$, leaving $\mathcal{E}$ untouched. For $x=y=2$, the extracted state is within $C\sqrt{\epsilon}$ of $\ket{\phi_\mathcal{A'B'}^+}\otimes\ket{\xi_\mathcal{A''B''E}}$ in vector norm.
\begin{theorem}
  Let $A_x$ acting on $\mathcal{H}_\mathcal{A}$ and $B_y$ acting on $\mathcal{H}_\mathcal{B}$ be Alice's and Bob's quantum binary observables, respectively, in a CHSH scenario such that $A_x^2\leqslant \I$ and $B_y^2\leqslant \I$ for $x,y\in\{0,1\}$. Let $\rho_\mathcal{AB}$ be an arbitrary density operator on $\mathcal{H}_\mathcal{A}\otimes\mathcal{H}_\mathcal{B}$, and let $\ket{\psi_\mathcal{ABE}}$ be any purification of $\rho_\mathcal{AB}$. If $\Tr(\rho_\mathcal{AB}\mathfrak{B})=2\sqrt{2}-\epsilon$, with $0\leqslant\epsilon\leqslant2\sqrt{2}-2$, then there exist auxiliary Hilbert spaces $\mathcal{H}_\mathcal{A''}$ and $\mathcal{H}_\mathcal{B''}$ and local isometries $U_\mathcal{A}:\mathcal{H}_\mathcal{A}\to\mathbb{C}^2\otimes\mathcal{H}_\mathcal{A''}$ and $U_\mathcal{B}:\mathcal{H}_\mathcal{B}\to\mathbb{C}^2\otimes\mathcal{H}_\mathcal{B''}$, with $V=U_\mathcal{A}\otimes U_\mathcal{B}\otimes\I_\mathcal{E}$, and a normalized auxiliary state $\ket{\xi_\mathcal{A''B''E}}$ such that, for every $x,y\in\{0,1,2\}$,
    \begin{align}
    \Big\|&V\left(A_x\otimes\widetilde{B}_y\otimes\I_\mathcal{E}\right)\ket{\psi_\mathcal{ABE}}\nonumber\\
    &-\left[(A_x'\otimes B_y')\ket{\phi_\mathcal{A'B'}^+}\right]\otimes\ket{\xi_\mathcal{A''B''E}}\Big\|\leqslant C\sqrt{\epsilon}.
    \end{align}
  Here, $\ket{\xi_\mathcal{A''B''E}}\in\mathcal{H}_\mathcal{A''}\otimes\mathcal{H}_\mathcal{B''}\otimes\mathcal{H}_\mathcal{E}$ may be entangled across the bipartition $\mathcal{A''B''|E}$, and $C=(75+107\sqrt{2})/2^\frac{3}{4}$. The same local isometries and normalized auxiliary vector work for all $x,y\in\{0,1,2\}$, and the reference observables are assigned as
  \begin{align}
    & A_0'=B_0'=\sigma_x,\\
    & A_1'=B_1'=\sigma_z,\\
    & A_2=B_2=A_2'=B_2'=\I,\qquad \widetilde{B}_2=\I_\mathcal{B},
  \end{align}
\end{theorem}
\begin{proof}
  Throughout this proof, we write $\ket{\psi}\coloneqq\ket{\psi_\mathcal{ABE}}$ and $\norm{v}\equiv\norm{\,\ket{v}}$. We understand every operator to act as the identity on unmentioned systems, including $\mathcal{E}$. For any operator $K_\mathcal{AB}$, the purification gives
  \begin{equation}
    \norm{\left(K_\mathcal{AB}\otimes\I_\mathcal{E}\right)\ket{\psi_\mathcal{ABE}}}^2
    =\Tr\!\left(\rho_\mathcal{AB}K_\mathcal{AB}^\dagger K_\mathcal{AB}\right).
  \end{equation}
  Moreover, for every Hermitian operator $T$ being regularized,
  \begin{equation}
    \sgn^*\!\left(T\otimes\I_\mathcal{E}\right)=\sgn^*(T)\otimes\I_\mathcal{E}.
  \end{equation}
  These identities allow us to carry out the regularization and subsequent norm estimates on the purification while leaving the purifying system untouched.

  \textit{(Step 1. Regularization of Observables)}  Because the observables may be non-unitary, the bounds presented in the seminal work of McKague et al.~\cite{McKague2012Oct} require modifications to account for the non-projectiveness of the measurements. For this reason, let us consider the regularization of Alice's observables via the modified sign function $\sgn^*$:
  \begin{equation}
    \mathbb{A}_x=\sgn^*\left(A_x\right).
  \end{equation}
  For a bounded Hermitian operator $T$, we define $\sgn^*(T)=\int_{\mathbb{R}}\sgn^*(\lambda)\,dE_T(\lambda)$ by spectral calculus, where $E_T$ is its spectral measure. We use eigenvector sums below for the discrete-spectrum case, with the same estimates obtained by spectral integration in general. Thus, if $A_0 = \sum_k \lambda_k \ket{v_k}\bra{v_k}$, then $\sgn^*(A_0) = \sum_k \sgn^*(\lambda_k)\ket{v_k}\bra{v_k}$, where
  \begin{equation}
    \sgn^*(x)=\left\{
      \begin{array}{ll}
        1,&\text{ if }x\geqslant 0\\
        -1,&\text{ if }x<0
      \end{array} \right. .
    \end{equation}
    Therefore, $\mathbb{A}_x^2=\I$ for $x\in\{0,1\}$. Because
    \begin{equation}
    \widetilde{B}_0 \coloneqq \frac{B_0+B_1}{\sqrt{2}}\quad\textrm{and}\qquad \widetilde{B}_1\coloneqq \frac{B_0-B_1}{\sqrt{2}},
    \end{equation}
    we directly regularize Bob's original tilted observables as
    \begin{equation}
      \begin{aligned}
        \widetilde{\mathbb{B}}_0&\coloneqq \sgn^*\left(B_0+B_1\right)=\sgn^*\left(\widetilde{B}_0\right),\\
        \widetilde{\mathbb{B}}_1&\coloneqq \sgn^*\left(B_0-B_1\right)=\sgn^*\left(\widetilde{B}_1\right).
      \end{aligned}
    \end{equation}
    Note that the factor $1/\sqrt{2}$ is absorbed by the modified sign function. Also, both $\widetilde{\mathbb{B}}_0$ and $\widetilde{\mathbb{B}}_1$ are Hermitian and unitary by construction, satisfying $\widetilde{\mathbb{B}}_y^2=\I$ for $y\in\{0,1\}$.

    \textit{(Step 2. Splitting the Error)} We explicitly use the local extraction maps constructed by McKague et al.~\cite[Fig.~1 and Eq.~(8)]{McKague2012Oct}. Let
    \begin{equation}
      X_A=\mathbb{A}_0,\qquad Z_A=\mathbb{A}_1,\qquad
      X_B=\widetilde{\mathbb{B}}_0,\qquad
      Z_B=\widetilde{\mathbb{B}}_1.
    \end{equation}
    For each party $j\in\{A,B\}$, define
    $P_{j,\pm}=(\I_j\pm Z_j)/2$ and choose
    $\mathcal{H}_{j''}=\mathcal{H}_j$. With the extracted qubit
    written first, the local map is
    \begin{equation}
      U_j\ket{v}
      =\ket{0}_{j'}\otimes P_{j,+}\ket{v}
      +\ket{1}_{j'}\otimes X_jP_{j,-}\ket{v},
      \qquad \ket{v}\in\mathcal{H}_j.
    \end{equation}
    Since $X_j$ and $Z_j$ are Hermitian unitaries,
    \begin{equation}
      \begin{aligned}
        U_j^\dagger U_j
        &=P_{j,+}^{\,2}+P_{j,-}X_j^\dagger X_jP_{j,-}\\
        &=P_{j,+}+P_{j,-}=\I_j.
      \end{aligned}
    \end{equation}
    Thus $U_j$ is an isometry. Replacing $X_j$ and $Z_j$
    by $X_j\otimes\I_\mathcal{E}$ and $Z_j\otimes\I_\mathcal{E}$ in this
    construction gives exactly $U_j\otimes\I_\mathcal{E}$.
    Consequently, identifying $U_A=U_\mathcal{A}$ and
    $U_B=U_\mathcal{B}$, the extraction acting on the purification is
    \begin{equation}
      V=U_A\otimes U_B\otimes\I_\mathcal{E},
    \end{equation}
    which leaves the external purifying system untouched.
    Avoiding Naimark dilation of the measurements does not exclude the fresh local ancillas used in the extraction isometries.

    In Step~4, we show that the regularized observables satisfy
      \begin{align}
        \norm{\{\mathbb{A}_0,\mathbb{A}_1\} \ket{\psi_\mathcal{ABE}} } &\leqslant 2\epsilon_1 \label{norm_anti_bbA0-bbA1},\\
        \norm{\{\widetilde{\mathbb{B}}_0,\widetilde{\mathbb{B}}_1\} \ket{\psi_\mathcal{ABE}}} &\leqslant 2\epsilon_1 \label{norm_anti_bbB0tilde-bbB1tilde},\\
        \|(\mathbb{A}_0 - \widetilde{\mathbb{B}}_0) \ket{\psi_\mathcal{ABE}}\| &\leqslant \epsilon_2 \label{norm_bbA0-bbB0},\\
        \|(\mathbb{A}_1 - \widetilde{\mathbb{B}}_1) \ket{\psi_\mathcal{ABE}}\| &\leqslant \epsilon_2. \label{norm_bbA1-bbB1}
      \end{align}
    We apply the constructive proof of Theorem 1 of McKague et al.~\cite{McKague2012Oct} to the purification across $\mathcal{A|BE}$ using the local isometries defined above. All operator relations and vector-norm estimates in that proof, including the normalization of the auxiliary vector, remain valid when the operators are extended by the identity on $\mathcal{E}$. Therefore, there exists a normalized vector $\ket{\xi_\mathcal{A''B''E}}$ such that
    \begin{equation}\label{theo_scarani}
      \begin{aligned}
        \Big\|
        &V(\mathbb{A}_x\otimes\widetilde{\mathbb{B}}_y\otimes\I_\mathcal{E})\ket{\psi_\mathcal{ABE}}\\
        &-[(A'_x\otimes B'_y)\ket{\phi_\mathcal{A'B'}^+}]\otimes\ket{\xi_\mathcal{A''B''E}}
        \Big\|
        \leqslant \frac{11\epsilon_1+5\epsilon_2}{2}
        \eqqcolon f_\mathrm{unitaries},
      \end{aligned}
    \end{equation}
    for all $x,y\in\{0,1,2\}$, where $A_0'=B_0'=\sigma_x$, $A_1'=B_1'=\sigma_z$, and the identity settings are $A_2'=B_2'=\mathbb{A}_2=\widetilde{\mathbb{B}}_2=\I$. The same isometries and auxiliary vector apply to all settings. The isometries depend only on the local observables and are independent of the choice of purification, whereas the auxiliary vector may depend on that choice. Keeping the purifier explicit introduces no additional error factor.

\begin{widetext}
    The local isometries defined above have domains $U_\mathcal{A}:\mathcal{H}_\mathcal{A}\to\mathbb{C}^2\otimes\mathcal{H}_\mathcal{A''}$ and $U_\mathcal{B}:\mathcal{H}_\mathcal{B}\to\mathbb{C}^2\otimes\mathcal{H}_\mathcal{B''}$, with $V=U_\mathcal{A}\otimes U_\mathcal{B}\otimes\I_\mathcal{E}$. We apply these same local isometries to the original non-regularized observables, sum zero, and use the triangle inequality and norm preservation under the isometry $V$, since $V^\dagger V=\I_\mathcal{ABE}$:
    \begin{align}
      & \norm{V\left(A_x\otimes \widetilde{B}_y\otimes\I_\mathcal{E}\right)\ket{\psi_\mathcal{ABE}}  -\left[(A_x'\otimes B_y')\ket{\phi_\mathcal{A'B'}^+}\right]\otimes \ket{\xi_\mathcal{A''B''E}}} \nonumber\\
      &\leqslant  \norm{V\left[\left(A_x\otimes \widetilde{B}_y-\mathbb{A}_x\otimes\widetilde{\mathbb{B}}_y\right)\otimes\I_\mathcal{E}\right]\ket{\psi_\mathcal{ABE}}}+\norm{V\left(\mathbb{A}_x\otimes\widetilde{\mathbb{B}}_y\otimes\I_\mathcal{E}\right)\ket{\psi_\mathcal{ABE}} -\left[(A_x'\otimes B_y')\ket{\phi_\mathcal{A'B'}^+}\right]\otimes \ket{\xi_\mathcal{A''B''E}}},\\
      &= \norm{\left[\left(A_x\otimes \widetilde{B}_y-\mathbb{A}_x\otimes\widetilde{\mathbb{B}}_y\right)\otimes\I_\mathcal{E}\right]\ket{\psi_\mathcal{ABE}}}+\norm{V\left(\mathbb{A}_x\otimes\widetilde{\mathbb{B}}_y\otimes\I_\mathcal{E}\right)\ket{\psi_\mathcal{ABE}} -\left[(A_x'\otimes B_y')\ket{\phi_\mathcal{A'B'}^+}\right]\otimes \ket{\xi_\mathcal{A''B''E}}},\\
      &\leqslant \Delta + f_\textrm{unitaries},\label{self-testing_non_unitary}
    \end{align}
    \end{widetext}
    where we define the error caused by the regularization as
    \begin{equation}\label{error_Delta}
      \Delta \coloneqq \norm{\left[\left(\mathbb{A}_x\otimes\widetilde{\mathbb{B}}_y-A_x\otimes\widetilde{B}_y\right)\otimes\I_\mathcal{E}\right]\ket{\psi_\mathcal{ABE}}},
    \end{equation}
    while the second norm in the triangle-inequality estimate is bounded by $f_\textrm{unitaries}$ in Eq.~\eqref{theo_scarani}.

    \textit{(Step 3. Bound for $\Delta$)} To bound the regularization error on Alice's side, we write $A_x=\sum_k \lambda_k\ket{v_k}\bra{v_k}$ in the discrete spectral notation, with eigenvalues satisfying $\lambda_k\in[-1,1]$ since $A_x^2\leqslant\I$. Following the definition of the modified sign function, $\mathbb{A}_x=\sum_k\sgn^*(\lambda_k)\ket{v_k}\bra{v_k}$, so
    \begin{equation}
      \mathbb{A}_x-A_x=\sum_k(\sgn^*(\lambda_k)-\lambda_k)\ket{v_k}\bra{v_k}
    \end{equation}
    For each eigenvalue, we have 
    \begin{equation}
    |\sgn^*(\lambda_k)-\lambda_k|=1-|\lambda_k|\leqslant 1-\lambda_k^2.
    \end{equation}
    Squaring this scalar inequality and applying spectral calculus yields the vector inequality below:
    \begin{equation}\label{difAx}
      \norm{(\mathbb{A}_x - A_x)\ket{\psi}} \leqslant \norm{(\I - A_x^2)\ket{\psi}}.
    \end{equation}

    For the purification $\ket{\psi}$ fixed above, the main premise of the theorem gives $\expval{\mathfrak{B}}=\Tr(\rho_\mathcal{AB}\mathfrak{B})=2\sqrt{2}-\epsilon$, where $0\leqslant\epsilon \leqslant 2\sqrt{2}-2\simeq 0.83$.

    To apply the Cauchy--Schwarz inequality, we introduce the following vectors:
    \begin{align}
      \ket{a}
      &\coloneqq
      \begin{pmatrix}
        (A_0\otimes\I)\ket{\psi}\\
        (A_1\otimes\I)\ket{\psi}
      \end{pmatrix},
      &
      \ket{b}
      &\coloneqq
      \begin{pmatrix}
        (\I\otimes(B_0+B_1))\ket{\psi}\\
        (\I\otimes(B_0-B_1))\ket{\psi}
      \end{pmatrix}.\label{CSvectors}
    \end{align}
    Since Alice's operators commute with Bob's operators, the inner product between these vectors is
    \begin{equation}
      \braket{a}{b}=\expval{\mathfrak{B}}.
    \end{equation}
    The Cauchy-Schwarz inequality therefore gives
    \begin{equation}\label{CSinequality}
      \left|\expval{\mathfrak{B}}\right|^2\leqslant \norm{a}^2\norm{b}^2.
    \end{equation}
    The two factors on the right-hand side are
    \begin{align}
      \norm{a}^2
      &=\expval{\left(A_0^2+A_1^2\right)\otimes\I},\\
      \norm{b}^2
      &=\expval{\I\otimes\left[\left(B_0+B_1\right)^2+\left(B_0-B_1\right)^2\right]}\nonumber\\
      &=2\expval{\I\otimes\left(B_0^2+B_1^2\right)}\leqslant 4.
    \end{align}
    From \eqref{CSinequality} and the above, we have
    \begin{equation}
      \expval{\left(A_0^2+A_1^2\right)\otimes\I}
      \geqslant \frac{\left|\expval{\mathfrak{B}}\right|^2}{4}.
    \end{equation}
    Repeating the same argument after writing $\mathfrak{B}=(A_0+A_1)\otimes B_0+(A_0-A_1)\otimes B_1$ gives the corresponding bound for Bob. Therefore, using $\expval{\mathfrak{B}}=2\sqrt{2}-\epsilon$, we obtain
    \begin{align}
      &\min\left\{\expval{\left(A_0^2+A_1^2\right)\otimes\I},
      \expval{\I\otimes\left(B_0^2+B_1^2\right)}\right\}\nonumber\\
      &\geqslant \frac{\left|\expval{\mathfrak{B}}\right|^2}{4}=2-\sqrt{2}\epsilon+\frac{\epsilon^2}{4}\geqslant 2-\sqrt{2}\epsilon.
    \end{align}
    Since each squared observable is bounded above by the identity, the inequality above implies, for every $i,j\in\{0,1\}$,
    \begin{align}
      \expval{A_i^2\otimes\I}&\geqslant 1-\sqrt{2}\epsilon,\\
      \expval{\I\otimes B_j^2}&\geqslant 1-\sqrt{2}\epsilon.
    \end{align}
    Now, let $X=A_i^2\otimes\I$ and $Y=\I\otimes B_j^2$. These operators are commuting positive contractions, and therefore
    \begin{equation}
      (\I-X)(\I-Y)\geqslant 0
      \quad\Longrightarrow\quad
      XY\geqslant X+Y-\I.
    \end{equation}
    Taking the expectation value and using the two inequalities above, we conclude that
    \begin{equation}
      1-2\sqrt{2}\epsilon\leqslant \expval{A_i^2\otimes B_j^2}.\label{CSboundend}
    \end{equation}
    for any pair $i,j\in\{0,1\}$. In particular, because $\expval{A_i^2\otimes B_j^2}\leqslant \expval{A_i^2}$ and $\expval{A_i^2\otimes B_j^2}\leqslant \expval{B_j^2}$, the above inequality implies that
    \begin{equation}\label{norm_Ai}
      1-2\sqrt{2}\epsilon \leqslant\expval{A_i^2} = \bra{\psi} A_i^\dagger A_i\ket{\psi} = \norm{A_i\ket{\psi}}^2\leqslant 1,
    \end{equation}
    with an identical expression for $B_i$. Note that we are going to omit tensor products and identities when there is no chance of confusion.

    From \eqref{norm_Ai}, we have
    \begin{equation}
      2\sqrt{2}\epsilon\geqslant \expval{\I-A_i^2}\geqslant \expval{\left(\I-A_i^2\right)^2}=\norm{\left(\I-A_i^2\right)\ket{\psi}}^2,
    \end{equation}
    which implies that
    \begin{equation}
      \norm{\left(\I-A_i^2\right)\ket{\psi}}\leqslant 2^{\frac{3}{4}}\sqrt{\epsilon},
    \end{equation}
    with an identical expression for $B_i$. Therefore, from \eqref{difAx} and the bound above, we obtain
    \begin{equation}\label{norm_bbA-A}
      \norm{(\mathbb{A}_x - A_x)\ket{\psi}} \leqslant 2^{\frac{3}{4}}\sqrt{\epsilon}.
    \end{equation}

    Now, let us calculate the error $\Delta$ defined by \eqref{error_Delta}. For that, let us sum zero in the following way:
    \begin{equation}\label{decomposition_Delta}
      A_x \otimes \widetilde{B}_y - \mathbb{A}_x \otimes \widetilde{\mathbb{B}}_y = (A_x - \mathbb{A}_x) \otimes \widetilde{B}_y + \mathbb{A}_x \otimes (\widetilde{B}_y - \widetilde{\mathbb{B}}_y).
    \end{equation}
    In general, for arbitrary operators $M$ and $N$, and quantum states $\ket{\psi}$, the Euclidean vector norm $\norm{\cdot}$ and the operator norm $\norm{\cdot}_\infty$ satisfy
    \begin{equation}
      \norm{MN\ket{\psi}}\leqslant \norm{M}_\infty  \norm{N\ket{\psi}}.
    \end{equation}
    By applying the above inequality to \eqref{decomposition_Delta} we obtain
    \begin{equation}\label{Delta}
      \Delta\leqslant \norm{\widetilde{B}_y}_\infty\norm{\left(A_x-\mathbb{A}_x\right)\ket{\psi}} + \norm{\mathbb{A}_x}_\infty\norm{\left(\widetilde{B}_y-\widetilde{\mathbb{B}}_y\right)\ket{\psi}}
    \end{equation}

    To bound the second term in \eqref{Delta}, we use the same regularization argument employed for Alice's observables, now applied to Bob's tilted observables:
    \begin{equation}
      \left|\lambda-\sgn^*(\lambda)\right|=\left||\lambda|-1\right|\leqslant \left|1-\lambda^2\right|.
    \end{equation}
    Consequently, spectral calculus gives
    \begin{equation}
      \left(\widetilde{B}_y-\widetilde{\mathbb{B}}_y\right)^2\leqslant\left(\I-\widetilde{B}_y^{\,2}\right)^2.
    \end{equation}
    Taking the expectation value in $\ket{\psi}$ and then the square root yields
    \begin{equation}
      \norm{\left(\widetilde{B}_y-\widetilde{\mathbb{B}}_y\right)\ket{\psi}}\leqslant \norm{\left(\I-\widetilde{B}_y^2\right)\ket{\psi}},\qquad y\in\{0,1\}.
    \end{equation}
    Therefore, it remains to upper bound $\norm{(\I-\widetilde{B}_y^2)\ket{\psi}}$. Using the definitions $\widetilde{B}_0=(B_0+B_1)/\sqrt{2}$ and $\widetilde{B}_1=(B_0-B_1)/\sqrt{2}$, we have
    \begin{align}
      \I-\widetilde{B}_0^2 &=\frac{1}{2}\left(\I-B_0^2\right) + \frac{1}{2}\left(\I-B_1^2 \right)-\frac{1}{2}\{B_0,B_1 \},\\
      \I-\widetilde{B}_1^2 &=\frac{1}{2}\left(\I-B_0^2\right) + \frac{1}{2}\left(\I-B_1^2 \right)+\frac{1}{2}\{B_0,B_1 \}.
    \end{align}
    Applying the triangle inequality to both cases yields
    \begin{align}
      \norm{\left(\widetilde{B}_y-\widetilde{\mathbb{B}}_y\right)\ket{\psi}}&\leqslant \frac{1}{2}\norm{\left(\I-B_0^2\right)\ket{\psi}}+\frac{1}{2}\norm{\left(\I-B_1^2\right)\ket{\psi}}\nonumber\\
      &+\frac{1}{2}\norm{\left\{B_0,B_1\right\}\ket{\psi}}, \qquad y\in\{0,1\}.\label{norm_Btilde-bbBtilde}
    \end{align}
    The first two terms in \eqref{norm_Btilde-bbBtilde} are controlled by the same estimate obtained from \eqref{norm_Ai} for Bob's observables, namely
    \begin{equation}\label{norm_1-B2}
      \norm{\left(\I-B_y^2\right)\ket{\psi}}\leqslant 2^\frac{3}{4}\sqrt{\epsilon}.
    \end{equation}
    It remains to bound the anti-commutator term in \eqref{norm_Btilde-bbBtilde}. For this, we rewrite it in terms of tilted observables:
    \begin{align}
      \{B_0,B_1\}&=\widetilde{B}_0^2-\widetilde{B}_1^2,\\
      &= A_0\widetilde{B}_0-\widetilde{B}_0\left(A_0-\widetilde{B}_0\right)-A_1\widetilde{B}_1+\widetilde{B}_1\left(A_1-\widetilde{B}_1\right).
    \end{align}
    Now, we can apply the triangle inequality to get
    \begin{align}
      \norm{\{B_0,B_1\}\ket{\psi}}&\leqslant \norm{\left(A_0\widetilde{B}_0-A_1\widetilde{B}_1\right)\ket{\psi}}\nonumber\\
      &+\norm{-\widetilde{B}_0\left(A_0-\widetilde{B}_0\right)\ket{\psi}+\widetilde{B}_1\left(A_1-\widetilde{B}_1\right)\ket{\psi}}.\label{norm_anticommutatorB0B1}
    \end{align}
    Let us start by dealing with the first norm on the right-hand side of the above inequality. To facilitate the manipulation of terms, let us employ the following notation:
    \begin{equation}
      \ket{v_0}=A_0\widetilde{B}_0\ket{\psi}\quad\text{and}\quad\ket{v_1}=A_1\widetilde{B}_1\ket{\psi}.
    \end{equation}
    Now, consider the parallelogram formula,
    \begin{equation}\label{paralelogramo}
      \norm{v_0-v_1}^2=2\norm{v_0}^2+2\norm{v_1}^2-\norm{v_0+v_1}^2,
    \end{equation}
    and the fact that
    \begin{equation}\label{notationv0+v1}
      \ket{v_0+v_1}\equiv\ket{v_0}+\ket{v_1}=\frac{1}{\sqrt{2}}\mathfrak{B}\ket{\psi}.
    \end{equation}
    By using the Cauchy-Schwarz inequality,
    \begin{equation}\label{cauchy-schwarz}
      |\braket{\psi}{\phi}|^2\leqslant \norm{\psi}^2 \norm{\phi}^2\quad \overset{\norm{\psi}^2=1}{\Longrightarrow}\quad\norm{\phi}^2\geqslant |\braket{\psi}{\phi}|^2,
    \end{equation}
    and \eqref{notationv0+v1}, we can see that
    \begin{equation}
      \norm{v_0+v_1}^2\geqslant \frac{1}{2}\left| \bra{\psi}{\mathfrak{B}}\ket{\psi}\right|^2=\frac{1}{2}\left(2\sqrt{2}-\epsilon\right)^2\geqslant 4-2\sqrt{2}\epsilon.
    \end{equation}
    In addition,
    \begin{align}
      \norm{v_0}^2+\norm{v_1}^2&=\bra{\psi} \widetilde{B}_0A_0^2\widetilde{B}_0+\widetilde{B}_1A_1^2\widetilde{B}_1\ket{\psi}\nonumber\\
      &\leqslant \bra{\psi}\widetilde{B}_0^2+\widetilde{B}_1^2 \ket{\psi}\nonumber\\
      &=\bra{\psi}B_0^2+B_1^2 \ket{\psi}\leqslant 2 \bra{\psi}\I\ket{\psi}=2.
    \end{align}
    Inserting the two inequalities above in \eqref{paralelogramo} results in
    \begin{equation}
      \norm{v_0-v_1}^2\leqslant 2\sqrt{2}\epsilon,
    \end{equation}
    which means that
    \begin{equation}\label{norm_dif}
      \norm{\left(A_0\widetilde{B}_0-A_1\widetilde{B}_1\right)\ket{\psi}}\leqslant 2^\frac{3}{4}\sqrt{\epsilon}.
    \end{equation}
    Now, let us work with the second term on the right-hand side of \eqref{norm_anticommutatorB0B1}. To proceed, let us denote
    \begin{equation}
      \ket{u_0}=\left(A_0-\widetilde{B}_0\right)\ket{\psi},\qquad\textrm{and}\qquad \ket{u_1}=\left(A_1-\widetilde{B}_1\right)\ket{\psi}.
    \end{equation}
    Consequently,
    \begin{align}
      \norm{u_0}^2+\norm{u_1}^2 &=\bra{\psi}\left(A_0-\widetilde{B}_0\right)^2+\left(A_1-\widetilde{B}_1\right)^2\ket{\psi},\\
      & =\bra{\psi}A_0^2+A_1^2+B_0^2+B_1^2-2\left(A_0\widetilde{B}_0+A_1\widetilde{B}_1\right) \ket{\psi},\\
      & \leqslant 4-\sqrt{2}\bra{\psi}\mathfrak{B}\ket{\psi},\\
      & = \sqrt{2}\epsilon.
    \end{align}
    Since $\norm{u_0}^2\leqslant\norm{u_0}^2+\norm{u_1}^2$ and $\norm{u_1}^2\leqslant\norm{u_0}^2+\norm{u_1}^2$, the above inequality implies that
    \begin{equation}\label{norm_Ai-Bitilde}
      \norm{\left(A_i-\widetilde{B}_i\right)\ket{\psi}}\leqslant 2^\frac{1}{4}\sqrt{\epsilon},\qquad i\in\{0,1\}.
    \end{equation}
    Now, let us define the following row and column matrices:
    \begin{equation}
      M\coloneqq
      \begin{pmatrix}
        -\widetilde{B}_0 & \widetilde{B}_1
      \end{pmatrix}\qquad\text{and}\qquad \ket{u}\coloneqq
      \begin{pmatrix}
        \ket{u_0}\\
        \ket{u_1}
      \end{pmatrix},
    \end{equation}
    such that
    \begin{equation}
      -\widetilde{B}_0\left(A_0-\widetilde{B}_0\right)\ket{\psi}+\widetilde{B}_1\left(A_1-\widetilde{B}_1\right)\ket{\psi}=M\ket{u}.
    \end{equation}
    Observe that
    \begin{equation}
      \norm{M\ket{u}}\leqslant\norm{M}_\infty\norm{u}\qquad\textrm{and}\qquad\norm{M}_\infty^2=\norm{MM^\dagger}_\infty.
    \end{equation}
    Therefore, from
    \begin{equation}
      MM^\dagger=\widetilde{B}_0^2+\widetilde{B}_1^2=B_0^2+B_1^2\leqslant 2\I
    \end{equation}
    and
    \begin{equation}
      \norm{u}=\sqrt{\norm{u_0}^2+\norm{u_1}^2}\leqslant 2^\frac{1}{4}\sqrt{\epsilon},
    \end{equation}
    we can conclude that
    \begin{equation}\label{normMv}
      \norm{-\widetilde{B}_0\left(A_0-\widetilde{B}_0\right)\ket{\psi}+\widetilde{B}_1\left(A_1-\widetilde{B}_1\right)\ket{\psi}}\leqslant 2^\frac{3}{4}\sqrt{\epsilon}.
    \end{equation}
    Finally, by inserting \eqref{normMv} and \eqref{norm_dif} into \eqref{norm_anticommutatorB0B1}, we obtain
    \begin{equation}\label{norm_anti_B0B1}
      \norm{\left\{B_0,B_1 \right\}\ket{\psi}}\leqslant 2^\frac{7}{4}\sqrt{\epsilon}
    \end{equation}
    From \eqref{norm_Btilde-bbBtilde}, \eqref{norm_1-B2}, and the above, we obtain
    \begin{equation}\label{norm_final_Btilde-bbBtilde}
      \norm{\left(\widetilde{B}_y-\widetilde{\mathbb{B}}_y\right)\ket{\psi}}\leqslant 2^\frac{7}{4}\sqrt{\epsilon}
    \end{equation}
    Since $\norm{\widetilde{B}_y}_\infty\leqslant \sqrt{2}$ and $\norm{\mathbb{A}_x}_\infty=1$, we can insert \eqref{norm_final_Btilde-bbBtilde} and \eqref{norm_bbA-A} into \eqref{Delta} to get the final error factor
    \begin{equation}\label{bound_Delta}
      \Delta \leqslant 2^\frac{1}{4}\left(2+2\sqrt{2}\right)\sqrt{\epsilon}.
    \end{equation}
    For the identity settings, Eqs.~\eqref{norm_bbA-A} and~\eqref{norm_final_Btilde-bbBtilde} give $\Delta\leqslant 2^{3/4}\sqrt{\epsilon}$ when $y=2$ and $x\in\{0,1\}$, and $\Delta\leqslant 2^{7/4}\sqrt{\epsilon}$ when $x=2$ and $y\in\{0,1\}$, respectively. For $x=y=2$, $\Delta=0$. Thus Eq.~\eqref{bound_Delta} holds for all $x,y\in\{0,1,2\}$.

    \textit{(Step 4. Bound for $f_\textrm{unitaries}$)} To obtain $f_\textrm{unitaries}$, we have to determine $\epsilon_1$ and $\epsilon_2$ to apply Theorem 1 of McKague et al. \cite{McKague2012Oct}, comprised in Eq. \eqref{theo_scarani}. Let us start by summing zero inside the norm of the left side in \eqref{norm_bbA0-bbB0}, that is,
    \begin{align}
      &\norm{(\mathbb{A}_0 - \widetilde{\mathbb{B}}_0) \ket{\psi}}\nonumber\\
      &= \norm{(\mathbb{A}_0 - A_0) \ket{\psi}+(A_0 - \widetilde{B}_0) \ket{\psi}+(\widetilde{B}_0 - \widetilde{\mathbb{B}}_0) \ket{\psi}}, \\
      &\leqslant \norm{(\mathbb{A}_0 - A_0) \ket{\psi}}+\norm{(A_0 - \widetilde{B}_0) \ket{\psi}}+\norm{(\widetilde{B}_0 - \widetilde{\mathbb{B}}_0) \ket{\psi}}.
    \end{align}
    We already successfully derived a bound for each one of the three terms on the right-hand side of the above inequality, and they are given by \eqref{norm_bbA-A}, \eqref{norm_Ai-Bitilde}, and \eqref{norm_final_Btilde-bbBtilde}, respectively. For that reason, we have
    \begin{equation}\label{epsilon_2}
      \epsilon_2 = (2^{3/4} + 2^{1/4} + 2^{7/4})\sqrt{\epsilon}=2^\frac{1}{4} \left(1 + 3\sqrt{2}\right)\sqrt{\epsilon}.
    \end{equation}
    Note that \eqref{norm_bbA1-bbB1} is bounded by the same value above, which allows us to write
    \begin{equation}\label{norm_bbA-bbBtilde}
      \norm{(\mathbb{A}_i - \widetilde{\mathbb{B}}_i) \ket{\psi}}\leqslant \epsilon_2 = 2^\frac{1}{4} \left(1 + 3\sqrt{2}\right)\sqrt{\epsilon},\qquad i\in\{0,1\}.
    \end{equation}

    Now, let us establish a bound for $\epsilon_1$ by analyzing \eqref{norm_anti_bbA0-bbA1}. First, note that by symmetry, Alice's non-regularized observables must obey a relation identical to \eqref{norm_anti_B0B1}, that is,
    \begin{equation}\label{norm_anti_A0A1}
      \norm{\left\{A_0,A_1 \right\}\ket{\psi}}\leqslant 2^\frac{7}{4}\sqrt{\epsilon}.
    \end{equation}
    Let us use the relation above to bound the effect of Alice's regularized anti-commutator by summing zero:
    \begin{align}
      \norm{\{\mathbb{A}_0,\mathbb{A}_1\} \ket{\psi} } & = \norm{\left(\mathbb{A}_0\mathbb{A}_1+\mathbb{A}_1\mathbb{A}_0\right)\ket{\psi}},\\
      &\leqslant \norm{\left\{A_0,A_1 \right\}\ket{\psi}}+\norm{\left( \mathbb{A}_0\mathbb{A}_1-A_0A_1 \right)\ket{\psi}}\nonumber\\
      &+\norm{\left( \mathbb{A}_1\mathbb{A}_0-A_1A_0 \right)\ket{\psi}}.\label{midstep_norm_bbA0bbA1}
    \end{align}
    On the right-hand side of the above inequality, the second term can be bounded using the triangle inequality as follows:
    \begin{align}
      &\norm{\left( \mathbb{A}_0\mathbb{A}_1-A_0A_1 \right)\ket{\psi}} \nonumber\\
      &\leqslant \norm{\mathbb{A}_0(\mathbb{A}_1 - A_1)\ket{\psi}} + \norm{(\mathbb{A}_0 - A_0)A_1\ket{\psi}},\\
      &\leqslant \norm{\mathbb{A}_0}_\infty \norm{(\mathbb{A}_1 - A_1)\ket{\psi}} + \norm{(\mathbb{A}_0 - A_0)A_1\ket{\psi}}.\label{3norms}
    \end{align}
    Since $\norm{\mathbb{A}_0}_\infty=1$ and the first term in the above inequality is bounded by \eqref{norm_bbA-A}, we still have to bound the last term. We have that
    \begin{align}
      &\norm{(\mathbb{A}_0 - A_0)A_1\ket{\psi}}\nonumber\\
      &\leqslant \norm{\left(\mathbb{A}_0 - A_0\right)\left(A_1 - \widetilde{B}_1\right)\ket{\psi}} + \norm{\left(\mathbb{A}_0 - A_0\right)\widetilde{B}_1\ket{\psi}},\\
      &\leqslant \norm{\mathbb{A}_0 - A_0}_\infty \norm{\left(A_1 - \widetilde{B}_1\right)\ket{\psi}} + \norm{\widetilde{B}_1}_\infty \norm{\left(\mathbb{A}_0 - A_0\right)\ket{\psi}}.
    \end{align}
    On the right-hand side of the above inequality, the four factors are bounded using $\norm{\mathbb{A}_0 - A_0}_\infty\leqslant 2$, \eqref{norm_Ai-Bitilde}, $\norm{\widetilde{B}_1}_\infty\leqslant \sqrt{2}$, and \eqref{norm_bbA-A}, respectively. Therefore, we have
    \begin{equation}
      \norm{(\mathbb{A}_0 - A_0)A_1\ket{\psi}}\leqslant 2^\frac{9}{4}\sqrt{\epsilon}.
    \end{equation}
    Back to \eqref{3norms}, we have
    \begin{equation}
      \norm{\left( \mathbb{A}_0\mathbb{A}_1-A_0A_1 \right)\ket{\psi}}\leqslant 1\cdot  2^{\frac{3}{4}}\sqrt{\epsilon} +  2^\frac{9}{4}\sqrt{\epsilon}=2^\frac{1}{4} \left(4 + \sqrt{2}\right)\sqrt{\epsilon}.
    \end{equation}
    An identical analysis will provide
    \begin{equation}
      \norm{\left( \mathbb{A}_1\mathbb{A}_0-A_1A_0 \right)\ket{\psi}}\leqslant2^\frac{1}{4} \left(4 + \sqrt{2}\right)\sqrt{\epsilon}.
    \end{equation}
    Finally, we can insert \eqref{norm_anti_A0A1} and the above two inequalities into \eqref{midstep_norm_bbA0bbA1} to obtain
    \begin{align}
      \norm{\{\mathbb{A}_0,\mathbb{A}_1\} \ket{\psi} }&\leqslant 2^\frac{7}{4}\sqrt{\epsilon} + 2\cdot  2^\frac{1}{4} \left(4 + \sqrt{2}\right)\sqrt{\epsilon}\nonumber\\
      &=2^\frac{9}{4}\left(2+\sqrt{2}\right)\sqrt{\epsilon}.\label{norm_anticommutator_bbA0bbA1}
    \end{align}
    To finish, let us bound Bob's tilted regularized anti-commutator $\norm{\left\{\widetilde{\mathbb{B}}_0,\widetilde{\mathbb{B}}_1 \right\}\ket{\psi}}$. First, note that we can write
    \begin{align}
      \norm{\left\{\widetilde{\mathbb{B}}_0,\widetilde{\mathbb{B}}_1 \right\}\ket{\psi}}&=\Big|\Big|\left(\widetilde{\mathbb{B}}_0\widetilde{\mathbb{B}}_1 - \mathbb{A}_1\mathbb{A}_0\right)\ket{\psi} + \left(\widetilde{\mathbb{B}}_1\widetilde{\mathbb{B}}_0 - \mathbb{A}_0\mathbb{A}_1\right)\ket{\psi}\nonumber\\
      &+ \left\{\mathbb{A}_0,\mathbb{A}_1\right\} \ket{\psi}\Big|\Big|.\label{anticommutatorBB-AA_mid}
    \end{align}
    The first term on the right-hand side above can be written as
    \begin{align}
      &\norm{\left(\widetilde{\mathbb{B}}_0\widetilde{\mathbb{B}}_1 - \mathbb{A}_1\mathbb{A}_0\right) \ket{\psi}}\nonumber\\
      &= \norm{\left(\widetilde{\mathbb{B}}_0\widetilde{\mathbb{B}}_1 - \widetilde{\mathbb{B}}_0\mathbb{A}_1 + \widetilde{\mathbb{B}}_0\mathbb{A}_1 - \mathbb{A}_1\mathbb{A}_0\right)\ket{\psi}},\\
      & \leqslant \norm{\widetilde{\mathbb{B}}_0}_\infty\norm{\left(\widetilde{\mathbb{B}}_1 - \mathbb{A}_1\right)\ket{\psi}} + \norm{\mathbb{A}_1}_\infty\norm{\left(\widetilde{\mathbb{B}}_0 - \mathbb{A}_0\right)\ket{\psi}}.
    \end{align}
    Because $\norm{\widetilde{\mathbb{B}}_0}_\infty=\norm{\mathbb{A}_1}_\infty=1$ and \eqref{norm_bbA-bbBtilde}, the above inequality becomes
    \begin{equation}
      \norm{\left(\widetilde{\mathbb{B}}_0 \widetilde{\mathbb{B}}_1 - \mathbb{A}_1\mathbb{A}_0\right)  \ket{\psi}}\leqslant 2^\frac{5}{4} \left(1 + 3\sqrt{2}\right)\sqrt{\epsilon}.
    \end{equation}
    A similar analysis provides
    \begin{equation}
      \norm{\left(\widetilde{\mathbb{B}}_1\widetilde{\mathbb{B}}_0 - \mathbb{A}_0\mathbb{A}_1\right) \ket{\psi}}\leqslant 2^\frac{5}{4} \left(1 + 3\sqrt{2}\right)\sqrt{\epsilon}.
    \end{equation}
    Now, we can insert the two inequalities above and \eqref{norm_anticommutator_bbA0bbA1} into \eqref{anticommutatorBB-AA_mid} to obtain
    \begin{align}
      \norm{\left\{\widetilde{\mathbb{B}}_0,\widetilde{\mathbb{B}}_1 \right\}\ket{\psi}}&\leqslant 2\cdot 2^\frac{5}{4} \left(1 + 3\sqrt{2}\right)\sqrt{\epsilon}+2^\frac{9}{4}\left(2+\sqrt{2}\right)\sqrt{\epsilon}\nonumber\\
      &= 2^\frac{9}{4}\left(3+4\sqrt{2}\right)\sqrt{\epsilon}.
    \end{align}
    Therefore, we may take
    \begin{align}
      \epsilon_1 &=\frac{1}{2}\max\left\{\norm{\left\{\mathbb{A}_0,\mathbb{A}_1 \right\}\ket{\psi}}, \norm{\left\{\widetilde{\mathbb{B}}_0,\widetilde{\mathbb{B}}_1 \right\}\ket{\psi}}\right\}\nonumber\\
      &\leqslant 2^\frac{5}{4}\left(3+4\sqrt{2}\right)\sqrt{\epsilon}.\label{epsilon_1}
    \end{align}
    Finally, from \eqref{epsilon_1} and \eqref{epsilon_2}, we obtain
    \begin{align}
      f_\mathrm{unitaries}=\frac{11\epsilon_1 + 5\epsilon_2}{2} \leqslant \frac{71+103\sqrt{2}}{2^\frac{3}{4}}\sqrt{\epsilon} \simeq 128.8 \sqrt{\epsilon}.\label{f_unitaries}
    \end{align}
    By inserting \eqref{bound_Delta} and \eqref{f_unitaries} into \eqref{self-testing_non_unitary} we obtain
    \begin{align}
    \Delta+f_{\mathrm{unitaries}} &\leqslant \left[ 2^{1/4}\left(2+2\sqrt{2}\right) +\frac{71+103\sqrt2}{2^{3/4}}\right]\sqrt{\epsilon}\nonumber\\
    &= \frac{75+107\sqrt2}{2^{3/4}}\sqrt{\epsilon},
    \end{align}
    which completes the proof.
\end{proof}

\section{Comparison with lifted projective self-tests}\label{app:comparison_constants}
We first derive bounds for the original observables and the individual measurement effects from our Theorem~\ref{theorem1}. Let $V$, $\ket{\psi}$, $\ket{\phi^+}$, and $\ket{\xi}$ be the isometry and states in that theorem, with subsystem labels on the states suppressed. Bob's original observables satisfy
\begin{equation}\label{eq:comparison_bob_decomposition}
  B_y=\frac{\widetilde{B}_0+(-1)^y\widetilde{B}_1}{\sqrt{2}},
  \qquad y\in\{0,1\}.
\end{equation}
By linearity, the corresponding extraction error is a linear combination of the two errors bounded in Theorem~\ref{theorem1}. The triangle inequality therefore gives
  \begin{align}
    &\Big\|V(A_x\otimes B_y\otimes\I_\mathcal{E})\ket{\psi}-\Big[\Big(A_x'\otimes
    \frac{\sigma_x+(-1)^y\sigma_z}{\sqrt{2}}\Big)
    \ket{\phi^+}\Big]\otimes\ket{\xi}\Big\|\nonumber\\
    &\leqslant\frac{C\sqrt{\epsilon}+C\sqrt{\epsilon}}{\sqrt{2}}=\sqrt{2}C\sqrt{\epsilon},\label{eq:comparison_bob_action}
  \end{align}
for $x\in\{0,1,2\}$ and $y\in\{0,1\}$. The same local isometries and normalized auxiliary vector apply to every setting. In particular, $x=2$ gives the bound for Bob's individual observable actions.

For the binary measurement effects,
\begin{equation}\label{eq:comparison_binary_effects}
  A_{a|x}=\frac{\I+aA_x}{2},
  \qquad
  B_{b|y}=\frac{\I+bB_y}{2},
\end{equation}
linearity and the triangle inequality bound each extraction error by half the sum of the errors for the state and the corresponding observable. Theorem~\ref{theorem1} gives $C\sqrt{\epsilon}$ for the state and Alice's observable actions, while Eq.~\eqref{eq:comparison_bob_action} gives $\sqrt{2}C\sqrt{\epsilon}$ for Bob's original observable actions. Consequently, Alice's effects have error at most $C\sqrt{\epsilon}$, and Bob's effects have error at most $(1+\sqrt{2})C\sqrt{\epsilon}/2$. A uniform bound for the state and all individual local effects is thus
\begin{equation}\label{eq:direct_effect_error}
  r_\mathrm{direct}(\epsilon)
  =\frac{1+\sqrt{2}}{2}C\sqrt{\epsilon}
  \simeq162.44\sqrt{\epsilon}.
\end{equation}
This bound holds for arbitrary mixed physical states, with common local isometries and a common normalized auxiliary vector.

We specialize the estimates of Bamps and Pironio~\cite{Bamps2015May} to CHSH, with tilt $\alpha=0$, $\theta=\mu=\pi/4$, $s=1$, and $c=0$ in their notation. Denote their quantity $\delta=\sqrt{2I_\alpha^\mathrm{max}\epsilon}$ by $d$, so that $d=2^{5/4}\sqrt{\epsilon}$. The remaining $\delta$ symbols below refer to their intermediate bounds. Their Eqs.~(A23)--(A26), (A38), (A49)--(A50), and~(A64) give
\begin{align}
  \delta_1=\frac{d}{\sqrt{2}},\quad
  \delta_4=(1+\sqrt{2})\delta_1,\quad
  \delta_a^\mathcal{A}=\left(1+\frac{1}{\sqrt{2}}\right)d,\label{eq:bp_intermediate_a}\\
  \delta_a^\mathcal{B}=\left(4+\frac{7}{\sqrt{2}}\right)d,\quad
  \bar{\delta}=\frac{5+14\sqrt{2}}{4}d.\label{eq:bp_intermediate_b}
\end{align}
Let $e_\mathrm{s}$, $e_\mathcal{A}$, and $e_\mathcal{B}$ denote their resulting upper bounds for the state and the individual original observable actions, taking the larger bound over each party's two settings. Their Eqs.~(A57) and~(A70)--(A72) yield
\begin{align}
  e_\mathrm{s}&=2\bar{\delta}
  =2^{1/4}(5+14\sqrt{2})\sqrt{\epsilon},\label{eq:bp_state_error}\\
  e_\mathcal{A}&=e_\mathrm{s}+2\delta_a^\mathcal{A}
  =2^{1/4}(9+16\sqrt{2})\sqrt{\epsilon},\label{eq:bp_alice_error}\\
  e_\mathcal{B}&=e_\mathrm{s}+\frac{\delta_1+\delta_4+2\delta_a^\mathcal{B}}{\sqrt{2}}
  =2^{1/4}(21+23\sqrt{2})\sqrt{\epsilon}.\label{eq:bp_bob_error}
\end{align}
All three estimates use the same isometries and normalized auxiliary vector. Since a binary effect is half the sum of the identity and the signed observable, a uniform bound for the state and individual local effects is
\begin{equation}\label{eq:bp_effect_error}
  r_\mathrm{BP}(\epsilon)
  =\max\left\{e_\mathrm{s},\frac{e_\mathrm{s}+e_\mathcal{A}}{2},
  \frac{e_\mathrm{s}+e_\mathcal{B}}{2}\right\}
  =\frac{26+37\sqrt{2}}{2^{3/4}}\sqrt{\epsilon}.
\end{equation}
For pure physical states, Theorem~4.3 of Baptista et al.~\cite{Baptista2025Nov} lifts the projective self-test to arbitrary POVMs with a factor of five in the error bound. Applying it to $r_\mathrm{BP}(\epsilon)$ gives
\begin{equation}\label{eq:lifted_effect_error}
  \begin{aligned}
    r_\mathrm{lift}(\epsilon)&=5r_\mathrm{BP}(\epsilon)\\
    &=\frac{5(26+37\sqrt{2})}{2^{3/4}}\sqrt{\epsilon}
    \simeq232.86\sqrt{\epsilon}.
  \end{aligned}
\end{equation}
Both $r_\mathrm{direct}$ and $r_\mathrm{lift}$ bound the state and individual local effects in the same vector-norm metric. The relative reduction of the prefactor is
\begin{equation}\label{eq:comparison_prefactor_reduction}
    1-\frac{r_\mathrm{direct}(\epsilon)}{r_\mathrm{lift}(\epsilon)}
    \simeq0.3024,\qquad \epsilon>0.
\end{equation}
Thus our prefactor is approximately $30\%$ smaller than that obtained from this particular composition of estimates. Neither set of constants is optimized in this comparison.

For mixed physical states, consider the CHSH game with uniformly sampled inputs. Its winning probability is $p_\mathrm{win}=1/2+\langle\mathfrak{B}\rangle/8$, so its deficit is $\delta=p_\mathrm{win}^{\max}-p_\mathrm{win}=\epsilon/8$. Let $r_\mathrm{pure}(\delta)$ bound the error for pure physical states and arbitrary POVMs in the same metric. The proof of Theorem~4.7 of Baptista et al.~\cite{Baptista2025Nov} supplies the mixed-state error bound
\begin{equation}\label{eq:generic_purity_lifting}
  r_\mathrm{mix}(\delta)=2r_\mathrm{pure}(\delta)
  +\sqrt{\frac{2[\delta+K r_\mathrm{pure}(\delta)]}{g}},
\end{equation}
where $K>0$ depends on the game and $g>0$ is the gap between the two largest eigenvalues of its ideal game operator. Specializing Lemma~4.13 of Ref.~\cite{Baptista2025Nov} to the uniformly sampled CHSH game gives the admissible value $K=8$. The ideal game operator $\I/2+\sqrt{2}(\sigma_x\otimes\sigma_x+\sigma_z\otimes\sigma_z)/8$ has eigenvalues $1/2+\sqrt{2}/4$, $1/2$ (twice), and $1/2-\sqrt{2}/4$, so $g=\sqrt{2}/4$.

Taking $r_\mathrm{pure}(\delta)=r_\mathrm{lift}(8\delta)$ yields an $\mathcal{O}(\epsilon^{1/4})$ bound. Our direct construction retains $\mathcal{O}(\sqrt{\epsilon})$ for mixed physical states and arbitrary binary POVMs. This comparison concerns the stated, general estimate for removing the purity assumption; it does not establish an unavoidable loss of exponent for other dilation or lifting arguments.
\end{document}